\documentclass[a4paper,12pt]{article}
\usepackage[utf8]{inputenc}
\usepackage{amssymb,amsthm,amsmath}

\newtheorem{theorem}{Theorem}[section]
\newtheorem{definition}[theorem]{Definition}
\newtheorem{lemma}[theorem]{Lemma}

\newtheorem{corollary}[theorem]{Corollary}
\newtheorem{remark}[theorem]{Remark}
\newtheorem{example}[theorem]{Example}
\newtheorem{assumption}[theorem]{Assumption}

\begin{document}
\title{Skorohod's representation theorem and optimal strategies for markets with frictions}
\author{
	Huy N. Chau\thanks{Alfr\'{e}d R\'{e}nyi Institute of Mathematics, Hungary Hungarian Academy of Sciences, Re\'altanoda utca 13-15, Budapest, Hungary
		(\texttt{chau@renyi.hu},
		\texttt{rasonyi@renyi.hu}).} 
	\and
	Mikl\'{o}s R\'{a}sonyi\footnotemark[1]
}
\maketitle

% REQUIRED
\begin{abstract}
 We prove the existence of optimal strategies for agents with cumulative prospect theory preferences who trade
 in a continuous-time illiquid market, transcending known results which pertained only to 
 risk-averse utility maximizers. 
 The arguments exploit an extension of Skorohod's representation
 theorem for tight sequences of probability measures. This method is applicable in a number
 of similar optimization problems.
\end{abstract}

Keywords: 	Optimization, non-concave utility, Skorohod's representation, illiquidity, market frictions.

\section{Introduction}
Optimal investment for an agent with given preferences has always been a core topic in mathematical finance.
Classical papers on the subject (\cite{merton,samuelson}) as well as most subsequent studies
neglected the presence of market frictions such as transaction costs, taxes and liquidity effects,
and they also stuck to the paradigm of a concave utility function expressing risk-aversion of the agent.

Non-concave preferences involving distorted probabilities emerged over the time, \cite{kt,quiggin,tk},
and incorporating frictions in the model led to mathematical settings that are
different from the classical one, see e.g. \cite{ac} and 
Chapter 3 of \cite{ks}.

Our purpose in the present paper is to prove existence theorems for optimal strategies in a general, 
continuous-time setting, following the footsteps of \cite{KS99,ks_exp,sch,sara,sara2,owen,bouchard,campi-owen,cs,
	pen14}. 
Just like \cite{bouchard,campi-owen,cs,pen14}, we wish to treat markets with friction. 
The essential novelty is that our method allows preferences that correspond
to possibly non-concave utility functions and may involve distorted probabilities. 
Such problems seem to be intractable with the usual techniques of convex duality and
arguments involving convex combinations, \cite{KS99,ks_exp,sch}.
We propose a method for establishing the existence of optimizers based on an extension of Skorohod's famous 
representation theorem, see Theorem \ref{altalanos} and Remark \ref{worth} below. 

The approach we present works for
preferences of a very general form and for various financial models. Here we
confine ourselves to the illiquid market of \cite{gr} and to preferences in the spirit of cumulative
prospect theory (CPT), see \cite{kt,tk};
this setting illustrates the power of the method fairly well.  
%Another paper (in preparation) will cover the case of markets with transaction costs, see Example \ref{m}. 
Further extensions are left for future research.
We also point out that our method seems flexible enough for applications to e.g. 
model uncertainty where expected utility is maximized 
in the worst-case sense over a set of probabilities, see Remark \ref{rob} below.

Optimal investment with CPT preferences concentrated almost exclusively on frictionless markets:
\cite{bg,hz} treated one-step models and found rather precise conditions for the existence of
optimal portfolio. The papers \cite{cr,jose} considered multistep models and proved that there
are optimal strategies when the investor is allowed to use a randomization (which is independent of the market). 
The present paper is similar to \cite{cr,jose} in the sense that we also allow randomization, see Assumption \ref{ghj} below.

Most continuous-time studies 
assumed a complete market: \cite{bkp,cp,reichlin} considered nonconcave utilities but no probability distortions; in \cite{jz} explicit solutions
were obtained under suitable assumptions, see also \cite{cd}; \cite{cd11} considered informational aspects of the problem while
\cite{rr,rrr} investigated well-posedness.  
Only a narrow class of 
incomplete markets have
been treated so far, \cite{reichlinthesis,rr,jose1}, using ad hoc techniques. Further problems of optimal control within CPT were treated in
\cite{jin2013greed,chang2015optimal,he2016path,he2015optimal} but these are connected to our setting only remotely.

We are aware of only
\cite{teemu} that treats markets with frictions and agents with nonconcave preferences. That paper
established a fairly general dynamic programming principle in a discrete time setting without 
probability distortions which is applicable to optimization problems in a wide range of market models.
The present paper seems to be the first continuous-time study involving CPT preferences and
market frictions at the same time.

In Section \ref{tetto} we present an extension of Skorohod's representation theorem from \cite{jakubowski}
and verify that it applies to our setting. In Section \ref{harom} we present a model of an
illiquid market. In Section \ref{ins}
we use the representation theorem of Section \ref{tetto} to construct optimal strategies in investment 
problems under liquidity constraints. Section \ref{j} sketches an alternative formulation for our results.
Finally, Section \ref{negy} collects some useful lemmas.

\section{A representation theorem}\label{tetto}

For a random variable $X$ on some probability space we denote by $\mathrm{Law}(X)$ its law. When
there might be an ambiguity about the probability space we use the notation $\mathrm{Law}_Q(X)$
for the law of $X$ under the probability $Q$.

We denote by $\mathcal{B}(Z)$ the Borel-field of a topological space $Z$.
%and by $C_b(Z)$
%the family of continuous and bounded functions on $Z$. A sequence of probability
%measures $\mu_k$, $k\in\mathbb{N}$ on $\mathcal{B}(Z)$ is said to converge weakly to a measure $\mu$
%if 
%$$
%\int_{Z}f(x)\mu_k(dx)\to \int_{Z} f(x)\mu(dx),\ k\to\infty,
%$$ 
%for all $f\in C_b(Z)$. 
A sequence of probabilities $\mu_k$, $k\in\mathbb{N}$ on $\mathcal{B}(Z)$ is said to be \emph{tight} if, 
for all $\varepsilon>0$,
there is a compact $K(\varepsilon)\subset Z$ such that, for all $k$, $\mu_k(Z\setminus K(\varepsilon))<
\varepsilon$. We first recall a remarkable result from \cite{jakubowski}.

\begin{theorem}\label{altalanos}
	Let $Z$ be a topological space such that there is a countable collection $f_i$, $i\in\mathbb{N}$
	of continuous, real-valued functions which separate points on $Z$. Let $\mu_k$, $k\in\mathbb{N}$
	be a tight sequence of measures on $\mathcal{B}(Z)$. Then there is a subsequence $k_j$, $j\in\mathbb{N}$ 
	and a probability space on which there exist $Z$-valued random variables $\xi$, $\xi_j$, 
	with $\mathrm{Law}(\xi_j)=
	\mu_{k_j}$, $j\in\mathbb{N}$ and $\xi_j\to\xi$ a.s., $j\to\infty$. 
	\hfill $\Box$
\end{theorem}

\begin{lemma}\label{fanta}
	Let $Z$ be a regular Hausdorff topological space such that there is an increasing sequence $A_n$,
	$n\in\mathbb{N}$ of closed subspaces of $Z$ which are separable metric spaces (under a suitable metric) and 
	$Z=\cup_{n\in\mathbb{N}} A_n$.
	Then there is a countable collection of continuous, real-valued functions which separate points on $Z$.
\end{lemma}
\begin{proof}
	Each $A_n$ is Lindel\"of hence so is $Z$. A Lindel\"of regular space is normal, so $Z$ is also a normal
	Hausdorff space. For each $n$,
	there is clearly a sequence $f^n_i$, $i\in\mathbb{N}$ of continuous real-valued functions on $A_n$
	which separate points on $A_n$. These can be extended in a continuous
	way to $Z$ by Tietze's theorem, for all $n$, $i$. Then the countable collection of extended functions 
	$f^n_i$, $i,n\in\mathbb{N}$ separates points on $Z$. 
\end{proof}

\begin{corollary}\label{fonti}
	Let $\mathbb{B}$ be a separable Banach space with
	dual $\mathbb{B}'$ equipped with the weak-star topology and let $M$ be a separable metric space.
	Then $Z:=\mathbb{B}'\times M$ satisfies the hypotheses of Lemma \ref{fanta}.
\end{corollary}
\begin{proof}
	Indeed, topological vector spaces and metric spaces 
	are both regular; $\mathbb{B}'$ as well as $M$ are clearly Hausdorff. So the product $Z$ is regular Hausdorff.
	Denote by $||\cdot||'$ the norm of $\mathbb{B}'$ and set
	$B_n:=\{x\in\mathbb{B}':\ ||x||'\leq n\}$.
	In the weak-star topology, any closed ball in the dual of a separable Banach space is metrizable
	and compact, so $A_n:=B_n\times M$ is a separable
	metric space, closed in the relative topology of $Z$, for each $n$.
\end{proof}

\begin{example}\label{lp} {\rm Fix $1<\beta<\infty$. Let $\gamma$ be defined by $1/\beta+1/\gamma=1$.
		Let $L^{\beta}:=L^{\beta}([0,1],\mathcal{B}([0,1]),Leb)$ denote the usual Banach space of 
		(equivalence classes of) $\beta$-integrable 
		functions on the unit interval. Let $T$ be $L^{\beta}$ equipped with the weak topology. 
		$L^{\beta}$ is the dual of the separable Banach space $L^{\gamma}$ and the weak topology on 
		$L^{\beta}$
		is precisely the weak-star topology in the duality $(L^{\gamma},L^{\beta})$. For any 
		separable metric space $M$, Theorem \ref{altalanos} applies to $Z:=T\times M$,
		by Lemma \ref{fanta} and Corollary \ref{fonti}.
		
		This topological space $Z$ 
		will be crucial in our study of optimal investment
		in illiquid markets as strategies will be represented by random elements in $T$ and
		a certain $M$ will code the information structure of the market,
		see Section \ref{harom} for details.}
\end{example}

\begin{example}\label{m}
	{\rm Consider $C[0,1]$, the separable Banach space
		of continuous functions on the unit interval (with the supremum norm).
		Let $\mathcal{M}$ denote the Banach space of finite signed measures on $\mathcal{B}([0,1])$,
		the dual space of $C[0,1]$.
		We take $T$ to be $\mathcal{M}$ equipped with the weak-star topology. 
		Again, Theorem \ref{altalanos} applies to $Z:=T\times M$ for any separable metric space $M$.
		
		The space $Z$ can be used in the treatment of optimal investment under transaction costs 
		where strategies correspond to random elements in $T$ and the price process
		is assumed continuous (i.e. it is a random element in $C[0,1]$). Due to the numerous technicalities,
		details will not be presented here.}
\end{example}

\begin{remark}
	{\rm For the moment, the case of frictionless markets is not accessible with our methods as 
		the construction of stochastic integrals is carried out in a filtration-dependent way and cannot be performed
		pathwise.} 
\end{remark}

\begin{remark}
	{\rm Note that neither of the spaces in Examples \ref{lp}, \ref{m} is metrizable so the well-known
		versions of Skorohod's representation theorem (see e.g. Lemma
		4.30 in \cite{k}) are not applicable to them. We also point out that topological spaces
		with a Skorohod representation property behave delicately: they are not known to be
		closed for topological products; counterexamples show that, even for a weakly convergent sequence
		of probabilities, Skorohod representation may only work for a subsequence, etc.
		We refer the interested reader to \cite{survey} for details.}
\end{remark}

\section{A model of an illiquid market}\label{harom}

We now recall a simple version of the market model in \cite{gr} where security prices depend on the trading 
speed. In that model, price impact is 
assumed \emph{instantaneous} (the activities of the small agent in consideration do not move prices permanently) 
and \emph{superlinear}, see Assumption \ref{modl} below. Superlinearity is
in accordance with empirical studies, see e.g. \cite{rama}.

We will assume throughout the paper that trading takes place continuously
in the time interval $[0,1]$. Let $(\Omega, \mathcal{F}, (\mathcal{F}_t)_{t\in [0,1]}, P)$ be a filtered 
probability space, where the filtration is complete and right continuous, $\mathcal{F}_0$ is trivial.
A \emph{process} $\psi$ on this space is an $\mathcal{F}\otimes \mathcal{B}([0,1])$-measurable
mapping on $\Omega\times [0,1]$. The notation $EX$ will refer to the expectation of the 
random variable $X$. If there is ambiguity about the probability space then $E_QX$ will denote
the expectation of $X$ under the probability $Q$. We denote by $1_A$ the indicator of a set $A$.

In the sequel we will need that the filtration is of a specific type and that the probability space 
is large enough.

\begin{assumption}\label{u}
	There exists a c\`adl\`ag $\mathbb{R}^m$-valued process $Y$
	with independent increments such that $\mathcal{F}_t$ is the $P$-completion of $\sigma(Y_u,\, 0\leq u\leq t)$,
	for $t\in [0,1]$.
\end{assumption}

\begin{assumption}\label{ghj}
	There exists a random variable $U$ that is uniformly distributed on $[0,1]$ and independent
	of $\mathcal{F}_1$. 
\end{assumption}

For $m\in\mathbb{N}$, we denote by $\mathcal{D}^m$ the space of $\mathbb{R}^m$-valued right-continuous 
functions with left-hand limits
on $[0,1]$, equipped with Skorohod's topology, see Chapter 3 of \cite{billingsley}.
%When $m=1$ we use the notation $\mathcal{D}$ instead of $\mathcal{D}^1$.

\begin{remark}\label{ipszilon}{\rm The Borel-field of $\mathcal{D}^m$ is generated by the coordinate 
		mappings $x\in\mathcal{D}^m\to x(t)\in\mathbb{R}^m$, $t\in [0,1]$, see Theorem 12.5 of 
		\cite{billingsley}. It follows that the function $\omega\in \Omega\to Y(\omega)\in\mathcal{D}^m$
		is a random variable and so is $\omega\in \Omega\to ^t Y(\omega)\in\mathcal{D}^m$, for
		all $t\in [0,1]$, where $^t Y$ is the process defined as
		$(^t Y)_u=Y_u 1_{[0,t)}+Y_t 1_{[t,1]}$, $u\in [0,1]$. 
		Furthermore, $\mathcal{F}_t=\sigma(Y_s,\, s\leq t)=\sigma(^t Y)$, for all $t\in [0,1]$.}
\end{remark}

Let us define the augmented filtration $\mathcal{G}_t:=\mathcal{F}_t\vee \sigma(U)$. 
Standard arguments (like Lemma 4.9 of \cite{hasan}) imply that $\mathcal{G}_t$, $t\in [0,1]$ also satisfies
the usual hypotheses of completeness and right-continuity.

The market consists of a riskless asset 
$S^0$ with price $S^0_t = 1$ for all $t \in [0,1]$ (``money account'') 
and a risky asset whose price $S$ is assumed to be an $\mathbb{R}$-valued c\`adl\`ag 
adapted process. (The extension of our results is straightforward to
the case of multiple risky assets.)

\begin{lemma}\label{measurability}
	There exists a measurable function $f : \mathcal{D}^m \to \mathcal{D}^1$ such that
	$S=f(Y)$. Furthermore, $^t S$ is measurable with respect to $\sigma(^t Y)$,
	for all $t\in [0,1]$, where $^t S$ is the process defined as
	$(^t S)_u=S_u 1_{[0,t)}+S_t 1_{[t,1]}$, $u\in [0,1]$.
\end{lemma}
\begin{proof} As $S$ is a c\`adl\`ag process, it is a $\mathcal{D}^1$-valued random variable,
	by the argument of Remark \ref{ipszilon}. The first statement now follows from Doob's lemma 
	(Lemma 1.13 of \cite{k}) since $S$ is $\sigma(Y)$-measurable.
	The second statement also follows as in Remark \ref{ipszilon}.
\end{proof}

\begin{definition}\label{defi_strategy}
	A feasible strategy is a process $\phi:\Omega\times \mathbb{R}_+\to\mathbb{R}$ such that 
	it is progressively measurable with respect to 
	$\mathcal{G}_t$, $t\in [0,1]$ and  
	$$
	\int\limits_0^1 {|\phi_t|dt} < +\infty,\mbox{ a.s.}
	$$
	We denote by $\mathcal{A}$ the set of all 
	feasible strategies. 
\end{definition}

\begin{remark}
	{\rm We indicate that Definition \ref{defi_strategy} slightly deviates from the corresponding 
		Definition 2.1 in \cite{gr}. In that
		paper strategies are assumed optional while here we only require progressive measurability. The latter class fits
		better the purposes of the present paper and the proofs of all the results we cite from \cite{gr} 
		(Lemma 3.4 and Theorem 5.1)
		go through without any modifications for the class of progressively measurable processes as well.
		
		The process $\phi$ represents the \emph{trading rate}. Assume that the initial positions in the money account and in the stock are
		$z^0$, $z^1$, respectively. For each $\phi \in \mathcal{A}$, we may define by
		$$\varphi_t := z^1+\int_0^t{\phi_u\,du}, \qquad t \in [0,1],$$
		the number of risky assets in the portfolio at time $t$.  
		
		If there were no liquidity effects, the self-financing
		condition would imply that the change of the portfolio value over $[0,t]$ is $\int_0^t \varphi_u dS_u$ (implicitly
		assuming that $S$ is a semimartingale). As the value of the stock position at $t$ is $\varphi_t S_t$,
		a heuristic integration by parts gives that the value at $t$ of the money account is 
		\begin{eqnarray}\nonumber z^0+\int_0^t \varphi_u dS_u-\phi_t S_t &=&\\
		\label{faym}
		z^0-\int_0^t S_u\, d\varphi_u &=& z^0-\int_0^t \phi_u S_u\, du. 
		\end{eqnarray}
		Notice that the last expression makes
		mathematical sense for any $\phi\in\mathcal{A}$ and for any c\`adl\`ag $S$.}
\end{remark}

We now add liquidity effects to our model by a function $G$ in such a way that $G_t(x)$
represents the ``penalty'' for trading at speed $x$ at time $t$. 
%We follow the assumptions given in \cite{gr}.
\begin{assumption}\label{modl}
	There is $\alpha>1$ and a continuous function $H:\mathbb{R}\to\mathbb{R}$, such that 
	$G_t(x)=g(S_t,x)$ with
	\begin{equation}\label{bg}
	g(s,x)=H(s)|x|^{\alpha}
	\end{equation}
	and $\inf_{t\in [0,T]} H(S_t)>0$ a.s.
	Furthermore, fix $1 < \beta < \alpha$ and assume 
	\begin{equation}\label{condition_H}
	{E} \int_0^1{H^{\beta/(\beta - \alpha)}(S_t)(1 + |S_t|)^{\beta \alpha / (\alpha - \beta)}}\, dt < \infty.
	\end{equation}
\end{assumption}

\begin{remark}
	{\rm Typical specifications are $G_t(x)=\lambda |x|^{\alpha}$ or $G_t(x)=\lambda S_t |x|^{\alpha}$ with some $\alpha>1$,
		$\lambda>0$, see e.g. \cite{doso}. The first one satisfies Assumption \ref{modl} whenever 
		$\int_0^T E|S_t|^{\beta \alpha / (\alpha - \beta)}\, dt < \infty$, the second one whenever $S$ is positive, has continuous 
		trajectories and $\int_0^T [E|S_t|^{\beta (\alpha-1) / (\alpha - \beta)}+E|S_t|^{-\beta / (\alpha - \beta)}]\, dt < \infty$.
		It would be possible to substantially relax both \eqref{bg} and \eqref{condition_H} and to allow dependence of $H$ on the whole trajectory of $S$ but this would lead to complications without enhancing the message of our paper, so
		we refrain from seeking greater generality. 
		%One could also consider permanent or transient 
		%price impact, see e.g.
		%\cite{gatheral-schied}.
	}
\end{remark}

\begin{definition}
	For a given strategy $\phi \in \mathcal{A}$ and an initial position $z \in \mathbb{R}^2$, 
	the positions at time $t \in [0,T]$ in the risky and riskless asset are defined as
	\begin{align}
	\tilde{X}_t(\phi) &:= z^1 + \int_0^t{\phi_udu},  \nonumber\\
	X_t(\phi) &:= z^0 - \int_0^t{\phi_uS_udu} - \int_0^t{G_u(\phi_u)du}, \label{X^0}
	\end{align}
	respectively (compare to \eqref{faym} above). 
\end{definition}

Note that $X_t(\phi)$ may take the value $-\infty$.
For simplicity, we assume from now on that $z^0=z^1=0$, the case 
of nonzero initial
positions is easily incorporated into the present setting.

Let $G^*$ be the Fenchel-Legendre conjugate of $G$,
\begin{equation}\label{csincsin}
G^*_t(y) := \sup_{x \in \mathbb{R}} (xy-G_t(x))=
\frac{\alpha-1}{\alpha}\alpha^{1/(1-\alpha)}H^{1/(1-\alpha)}(S_t)|y|^{\alpha/(\alpha-1)},
\end{equation}
as an elementary calculation shows. From (\ref{X^0}), under Assumption \ref{modl} one has 
\begin{equation}\label{market_bound}
X_1(\phi) \le B:=\int_0^1{G^*_t(-S_t)dt} < \infty,\mbox{ a.s.},
\end{equation}
see Lemma 3.1 of \cite{gr}.
We call $B$ the \emph{market bound} as it dominates the terminal money account position of 
any feasible portfolio.

\section{Optimal investments}\label{ins}

For $x\in\mathbb{R}$ we denote $x^+:=\max\{x,0\}$, $x^-:=\max\{-x,0\}$.
Let $u_+,u_-: \mathbb{R}_+ \to \mathbb{R}_+$ be continuous, increasing functions such that $u_{\pm}(0)=0$.
Let $w_+,w_-: [0,1]\to [0,1]$ be continuous with $w_{\pm}(0)=0$, $w_{\pm}(1)=1$. Functions
$u_{\pm}$ express the agent's attitude towards gains and losses while $w_{\pm}$ are
functions distorting the probabilities of events, see \cite{tk}, \cite{cr}.

We define, for any random variable $X\geq 0$,
\begin{eqnarray}\nonumber
V_+(X):=\int_0^{\infty} w_+\left(P\left(
u_+\left(X\right)\geq y
\right)\right)dy,
\end{eqnarray}
and
\begin{eqnarray}\nonumber
V_-(X):=\int_0^{\infty} w_-\left(P\left(
u_-\left(X\right)\geq y
\right)\right)dy.
\end{eqnarray}
For each real-valued random variable $X$ with $V_-(X^-)<\infty$ we set 
\begin{equation}\nonumber
V(X):=V_+(X^+)-V_-(X^-). 
\end{equation}

Let $W$ be an $\mathcal{F}_1$-measurable random variable representing a benchmark for the agent 
in consideration. For example, $W$ can be the value of an index
or of the portfolio of a rival at time $1$ which serves as a reference point for our investor.
The quantity $V(X-W)$ expresses the satisfaction of an agent with CPT preferences when (s)he
receives a random amount $X$, see \cite{jz,cr} for more detailed discussions. 
Positive $X-W$ means outperforming a benchmark, negative $X-W$
means falling short of it. Doob's theorem implies that there is a measurable 
$\ell:\mathcal{D}^m\to \mathbb{R}$ such that
$W=\ell(Y)$.

Let us define $\mathcal{A}':=\{\phi\in\mathcal{A}:\,\tilde{X}_1=0,\ V_-([X_1(\phi)-W]^-)<\infty\}$.
For each $\phi\in\mathcal{A}'$ the position in the risky asset is liquidated by the terminal
date $1$ and the utility functional $V$ is well-defined for the value of the money account
at $1$ minus the benchmark. We aim to find an optimal investment strategy, i.e.  $\phi^{\dagger}\in\mathcal{A}'$ with
\begin{equation*}\label{problem}
V(X_1(\phi^{\dagger})-W)=\sup_{\phi\in\mathcal{A}'}V(X_1(\phi)-W).
\end{equation*}

\begin{remark}
	{\rm Note that if $w_{\pm}(p)=p$ (that is, there is no distortion) then we have 
		$V(X)=Eu(X)$ where $u(x)=u_+(x)$, $x\geq 0$ and $u(x)=-u_-(-x)$ for $x<0$.
		This shows that the above setting generalizes the well-known expected utility
		framework, see e.g. \cite{KS99,ks_exp,sch,sara,sara2}.
		
		It would be possible to prove analogues of Theorem \ref{main} below
		for other types of objectives e.g. the performance measures of 
		\cite{cherny}. We stress that the main purpose of the present
		paper is to demonstrate a useful method and not to explore all
		possible ramifications.}
\end{remark}

\begin{assumption}\label{bou}
	We assume that $V_+([B-W]^+)<\infty$ and $EW^+<\infty$. Furthermore, there exist $0<\delta_2<\delta_1$ such that
	\begin{equation}\label{mukk}
	u_-(x)\geq c_1 x^{\delta_1}-c_2
	\end{equation}
	and 
	\begin{equation}\label{mukk1}
	w_-(p)\geq c_3 p^{\delta_2},    
	\end{equation}
	with some
	constants $c_1,c_2,c_3>0$ and for all $x\in\mathbb{R}_+$, $p\in [0,1]$.
\end{assumption}

\begin{remark} {\rm $V_+([B-W]^+)<\infty$, $EW^+<\infty$ are integrability conditions that are easy to verify in concrete
		situations. Specifications of $u_-$, $w_-$ satisfy  
		\eqref{mukk}, \eqref{mukk1} with some $\delta_1,\delta_2>0$ quite often,
		going back to \cite{tk}. 
		It was shown in \cite{rr} that in a frictionless Black-Scholes market $\delta_1>\delta_2$ is
		necessary for well-posedness of \eqref{problem}. Hence the conditions of
		Assumption \ref{bou} are rather natural. If we assumed $u_+$ bounded above, we could substantially
		relax \eqref{mukk} and \eqref{mukk1} along the lines of \cite{rrr}.}
\end{remark}

For comparisons with Theorem \ref{main} below, we recall a consequence of Theorem 5.1 in \cite{gr}.

\begin{theorem}\label{utility} Let Assumption \ref{modl} be in vigour, let $u:\mathbb{R}\to\mathbb{R}$ be concave and nondecreasing, 
	and let $E|u(B-W)|<\infty$ hold. If $\mathcal{A}^{\circ}\neq\emptyset$ then there is $\phi^{\dagger}\in\mathcal{A}^{\circ}$
	such that
	\[
	Eu(X_1(\phi^{\dagger})-W)=\sup_{\phi\in\mathcal{A}^{\circ}}Eu(X_1(\phi)-W),
	\]  
	where $\mathcal{A}^{\circ}=\{\phi\in \mathcal{A}:\ \tilde{X}_1(\phi)=0,\ E(u(X_1(\phi)))^-<\infty\}$.
\end{theorem}
\begin{proof}
	Assumptions 2.2 and 2.3 of \cite{gr} hold by Assumption \ref{modl} above hence Theorem 5.1
	of \cite{gr} applies. We remark that, regrettably, the condition $\mathcal{A}^{\circ}\neq\emptyset$ is missing
	from the statement of Theorem 5.1 of \cite{gr} though it is clearly necessary. Here we
	publish a corrected statement.
\end{proof}

The next theorem is the main result of the present paper which extends Theorem \ref{utility}
to a much broader family of preferences.

\begin{theorem}\label{main} Let Assumptions \ref{u}, \ref{ghj}, \ref{modl} and \ref{bou} be in vigour.
	If $\mathcal{A}'\neq\emptyset$ then there exists $\phi^{\dagger}\in
	\mathcal{A}'$ such that 
	$$V(X_1(\phi^{\dagger})-W)=\sup_{\phi\in\mathcal{A}'}V(X_1(\phi)-W).$$
\end{theorem}
\begin{proof} 
	We provide a quick overview of the main steps in our argument. Taking an optimizer sequence we
	show that their laws form a tight sequence (on $L^{\beta}$ with the weak topology). Then we invoke
	results of Section \ref{tetto} to realize (on another probability space) a sequence whose members have the same laws but which converge
	almost surely. Using convex combinations coming from the theorem of Koml\'os we can show that the limit is an optimizer. However, we have
	to construct this optimizer on the original space as well so we use $U$ and rely on the usual construction of a random variable which
	has a given joint law with another, fixed random variable, see Lemma \ref{transfer}.

	Let us take $\phi_n\in\mathcal{A}'$, $n\in\mathbb{N}$ such that 
	$$
	V(X_1(\phi_n)-W)\to\sup_{\phi\in\mathcal{A}'}V(X_1(\phi)-W),\ n\to\infty.
	$$
	Recall that $\gamma$ denotes the conjugate number of $\beta$, see Example \ref{lp} above.
	We consider the space $L^{\beta}$ as defined in Example \ref{lp} above, equipped with the weak topology.
	We intend to use Corollary \ref{fonti}, Lemma \ref{fanta} and Theorem \ref{altalanos}
	with the choice $\mathbb{B}:=L^{\gamma}$ (then $\mathbb{B}'=L^{\beta}$)
	and $M:=\mathcal{D}^m\times (\mathbb{R} \cup \{-\infty\})$, 
	$\mu_n:=\mathrm{Law}(\phi_n,Y,X_1(\phi_n))$.
	
	First we show that $\phi_n:\Omega\to L^{\beta}$ is measurable when $L^{\beta}$ is equipped with
	the Borel field of the weak topology. It clearly suffices to show that, for all 
	$\mathcal{G}_1\otimes\mathcal{B}([0,1])$-measurable $\zeta$ with $\int_0^1 |\zeta_t(\omega)|^{\beta}\, dt<\infty$, 
	$\omega\to \int_0^1 \zeta(t)(\omega)q(t)dt$ is measurable for all $q\in L^{\gamma}$. Approximating $\zeta$
	by step functions and using linearity of the integral, it is enough to show this for 
	$\zeta:=1_K$ where $K\in\mathcal{G}_1\otimes\mathcal{B}([0,1])$. A monotone class argument
	reduces this to the case where $K=A\times B$ with $A\in\mathcal{G}_1$ and $B\in\mathcal{B}([0,1])$. 
	But then
	$$
	\int_0^1 \zeta(t)(\omega)q(t)dt=1_A(\omega) \int_B q(t)\, dt,
	$$
	which is trivially $\mathcal{G}_1$-measurable.
	
	$X_1(\phi_n)$ is a ($\mathbb{R}\cup\{-\infty\}$-valued) random variable by Lemmata \ref{con} and \ref{cop}.
	Finally, $Y:\Omega\to \mathcal{D}$ is measurable, see Remark \ref{ipszilon} above. This clearly
	implies the measurability of the triplet $(\phi_n,Y,X_1(\phi_n))$.
	
	\eqref{market_bound} implies that
	$\mathrm{Law}(X_1(\phi_n))$ is a tight sequence in $\mathbb{R}\cup \{-\infty\}$. Clearly, 
	$$\inf_n V(X_1(\phi_n)-W)>-\infty$$ so necessarily $\sup_n V_-([X_1(\phi_n)-W]^-)< \infty$,
	by \eqref{market_bound} and by $V_+([B-W]^+) < \infty$ in Assumption \ref{bou}.
	
	Lemma 3.12 of \cite{rr} (with the choice $s:=1$, $a:=\delta_2$, $b:=\delta_1$) implies that also 
	$\sup_n E(X_1(\phi_n) - W)^- <\infty$. 
	By the proof of Lemma 3.4 of \cite{gr} and by $E(X_1(\phi_n)-W)^-\leq E(X_1(\phi_n))^- + EW^+$,
	we get that 
	\begin{eqnarray}\label{barha}
	E\int_0^1 |\phi_n(t)|^{\beta}(1 + |S(t)|)^{\beta} dt\leq 
	E(X_1(\phi_n))^- + EW^+ &+&\\ 
	\nonumber 2^{\beta/(\alpha-\beta)}E\int_0^1
	{H^{\beta/(\beta - \alpha)}(S_t)(1 + |S_t|)^{\beta \alpha / (\alpha - \beta)}}\, dt+1 &=:& C<\infty,
	\end{eqnarray}
	by Assumptions \ref{modl} and \ref{bou}. As $C$ is independent of $n$, Markov's inequality implies
	$$
	P\left(\int_0^1 |\phi_n(t)|^{\beta}(1 + |S(t)|)^{\beta} dt\geq r\right)\leq C/r,
	$$
	for all $r>0$. Noting that closed balls of $L^{\beta}$ around the origin are weakly compact by the Banach-Alaoglu theorem (since $L^{\beta}$ is a reflexive
	Banach space), we get that $\mathrm{Law}(\phi_n)$, $n\in\mathbb{N}$ is a tight sequence of probabilities on 
	$\mathcal{B}(L^{\beta})$. Finally, as $Y$ takes values in a Polish space,
	$\mathrm{Law}(Y)$ is tight. It follows that $\mu_n$ is tight on $\mathcal{B}(L^{\beta}\times \mathcal{D}^m\times
	(\mathbb{R}\cup\{-\infty\}))$.
	
	Now apply Corollary \ref{fonti}, Lemma \ref{fanta} and Theorem \ref{altalanos} to get a probability space 
	$(O,\mathcal{H},Q)$ and
	$L^{\beta}\times \mathcal{D}\times (\mathbb{R} \cup \{-\infty\})$-valued random variables 
	$(\tilde{\phi}_n,Y_n,X_n)$ that converge
	a.s. to $(\phi^*,Y^*,X^*)$ along a subsequence (for which we keep the same notation) and 
	$\mathrm{Law}_Q(\tilde{\phi}_n,Y_n,X_n)=\mathrm{Law}(\phi_n,Y,X_1(\phi_n))$, $n\in\mathbb{N}$. 
	Passing to a further
	subsequence, we may and will assume $S_n:=f(Y_n)\to S^*:=f(Y^*)$ a.s. in $\mathcal{D}^1$, 
	by Lemmata \ref{measurability}, \ref{fonction} and by the fact that each $Y_n$
	has the same law (on $\mathcal{D}^m$). Analogously, we may and will assume $W_n=\ell(Y_n)\to W^*:=\ell(Y^*)$
	a.s. in $\mathbb{R}$. 
	By the argument of Lemma \ref{versi}, we may assume that $\tilde{\phi}_n$ can be identified with a
	$\mathcal{H}\otimes\mathcal{B}([0,1])$-measurable process.
	
	Let us define the analogue of the functionals $V_{\pm}$, $V$, for real-valued random variables
	$X$ on $(O,\mathcal{H},Q)$.
	\begin{eqnarray*}
		V_+^Q(X):=\int_0^{\infty} w_+\left(Q\left(
		u_+\left(X\right)\geq y
		\right)\right)dy,
	\end{eqnarray*}
	and
	\begin{eqnarray*}
		V_-^Q(X):=\int_0^{\infty} w_-\left(Q\left(
		u_-\left(X\right)\geq y
		\right)\right)dy.
	\end{eqnarray*}
	For each $X$ with $V_-^Q(X^-)<\infty$ we set 
	\begin{equation*}
	V^Q(X):=V^Q_+(X^+)-V^Q_-(X^-). 
	\end{equation*}
	
	Define 
	$$
	B_n:=\int_0^1\frac{\alpha-1}{\alpha}\alpha^{1/(1-\alpha)}H^{1/(1-\alpha)}(S^n_t)|S^n_t|^{\alpha/(\alpha-1)}dt,
	$$ 
	see \eqref{csincsin} and \eqref{market_bound}.
	The family $(B_n,W_n)$, $n\in\mathbb{N}$ has the same law under $Q$ (that of $(B,W)$ under $P$) hence the family
	of functions \begin{equation}\label{jaffa}t\to 
	w_+(Q(u_+([B_n-W_n]^+)\geq t)),\ n\in\mathbb{N},  
	             \end{equation}
        is uniformly integrable 
	(with respect to the Lebesgue measure on $\mathbb{R}_+$), by Assumption \ref{bou}. As for all $n$,
	$$
	-\int_0^1 \tilde{\phi}_n(t)S_n(t)\, dt-\int_0^1 H(S_n(t))|\tilde{\phi}_n(t)|^{\alpha}\, dt\leq B_n\mbox{ a.s.},
	$$
	uniform integrability of \eqref{jaffa} and 
	Fatou's lemma imply that $$V^Q(X^*-W^*)\geq \limsup_n V^Q(X_n-W_n),$$ so 
	$V^Q(X^*-W^*)\geq \sup_{\phi\in\mathcal{A}'} V(X_1(\phi)-W)$, in particular,
	$V^Q(X^*-W^*)>-\infty$.
	
	With the functional $F$ defined in 
	Lemma \ref{con}, we have, a.s., 
	$\int_0^1 \tilde{\phi}_n(t) S_n(t)\, dt=F(\tilde{\phi}_n,S_n)$ and hence, by Lemma \ref{con}, 
	\begin{equation}\label{mr}
	\int_0^1 \phi^*(t) S^*(t)\, dt=\lim_n
	\int_0^1 \tilde{\phi}_n(t) S_n(t)\, dt.
	\end{equation}
	It is also clear that \begin{equation}\label{akkad}
	\int_0^1 \phi^*(t)\, dt=0
	\end{equation}
 since $\tilde{\phi}_n$ tends to $\phi^*$
	a.s. weakly in $L^{\beta}$.
	
	From the almost sure convergence of $\tilde{\phi}_n$ to $\phi^*$ we get that, for almost every $\omega\in O$,
	\begin{equation}\label{borni}
	\sup_n \int_0^1 |\tilde{\phi}_n(t)(\omega)|^{\beta} dt<\infty
	\end{equation}
	(since a weakly convergent sequence in $L^{\beta}$ is weakly bounded hence
	also norm-bounded). A fortiori, $\sup_n \int_0^1 |\tilde{\phi}_n(t)(\omega)| dt<\infty$ a.s.
	
	Applying Lemma \ref{komlos} on the probability space 
	$(O\times [0,1],\mathcal{H}\otimes \mathcal{B}([0,1]),
	Q\otimes Leb)$ we get a subsequence (still denoted by $n$) such that
	$\widehat{\phi}_N:=\frac{1}{N}\sum_{n=1}^{N} \tilde{\phi}_n$ converge to
	some $\phi^{\circ}$ $P\otimes Leb$-a.s., $N\to\infty$, and, by \eqref{borni}, also 
	almost surely in the norm of $L^{1}$. This implies 
	convergence in the weak topology of $L^{1}$. As $\tilde{\phi}_n$ and hence also
	$\widehat{\phi}_n$ converge to $\phi^*$ in the weak topology of $L^{\beta}$ and
	thus also in the weak topology of $L^1$, we get that
	$\phi^{\circ}=\phi^*$
	necessarily, $P\otimes Leb$-a.s. hence we may and will use $\phi^{\circ}$ as a version
	of $\phi^*$ in what follows; $\phi^{\circ}$ is an $\mathcal{H} \times \mathcal{B}([0,1])$-measurable process. 
	
	Continuity of $H$ implies $H(S_n(t))\to H(S^*(t))>0$ a.s. so Fatou's lemma and convexity of $x\to |x|^{\alpha}$
	lead to
	\begin{eqnarray*}
	-X^*-\int_0^1 S^*(t)\phi^*(t)\, dt &=&\\
	\lim_{N} \frac{1}{N}\sum_{n=1}^N\int_0^1 H(S_n(t))|\tilde{\phi}_n(t)|^{\alpha}\, dt &\geq&\\
	\int_0^1 \liminf_N \frac{1}{N}\sum_{n=1}^N H(S_n(t))|\tilde{\phi}_n(t)|^{\alpha}\, dt &=&\\
	\int_0^1 \liminf_N \frac{1}{N}\sum_{n=1}^N H(S^*(t))|\tilde{\phi}_n(t)|^{\alpha}\, dt &\geq&\\
	\int_0^1 H(S^*(t))\liminf_N |\widehat{\phi}_N(t)|^{\alpha}\, dt &=&\\
	\int_0^1 H(S^*(t))|\phi^*(t)|^{\alpha}\, dt.
	\end{eqnarray*}

	%Fatou's lemma, continuity of $H$  and convexity of $x\to |x|^{\alpha}$ imply that
	%\begin{eqnarray*}
%		\int_0^1 -H(S(t))|\phi^*(t)|^{\alpha}\, dt &\geq& \limsup_n \int_0^1 -H(S_n(t))|\widehat{\phi_n}(t)|^{\alpha}\, dt\\
%		&\geq& \limsup_N 
%		\frac{1}{N}
%		\sum_{i=1}^{N} \int_0^1 -H(S_n(t))|\tilde{\phi}_n(t)|^{\alpha}\, dt. 
%	\end{eqnarray*}
%	Using \eqref{mr}, 

It follows that 
	$-\int_0^1 S^*(t)\phi^*(t)\, dt-\int_0^1 H(S^*(t))|\phi^*(t)|^{\alpha}\, dt\geq X^*$ so 
	\begin{equation}\label{hoarse}
	V^Q\left(-\int_0^1 S^*(t)\phi^*(t)\, dt-\int_0^1 H(S^*(t))|\phi^*(t)|^{\alpha}\, dt-W^*\right)
	\geq \sup_{\phi\in\mathcal{A}'} V(X_1(\phi)-W).
	\end{equation}
	
	Let us invoke Lemma \ref{transfer} with the choice $\tilde{\phi}:=\phi^*$, $\tilde{H}:=Y^*$ and
	$H:=Y$. We get a $\mathcal{G}_1$-measurable random element
	$\phi^{\ddagger}:=\phi\in L^{\beta}$ satisfying $\mathrm{Law}(\phi^{\ddagger},Y)=\mathrm{Law}_Q
	(\phi^*,Y^*)$. Let us fix $0\le t < u \le 1$. We recall that 
	$\phi_n 1_{[0,t]}$ is independent from ${Y}(u) - {Y}(t)$, or equivalently, 
	$$
	\mathrm{Law}({\phi}_n 1_{[0,t]},{Y}(u) - {Y}(t)) = \mathrm{Law}({\phi}_n 1_{[0,t]}) \otimes 
	\mathrm{Law}({Y}(u) - {Y}(t)).$$
	By construction, $\mathrm{Law}(\phi_n 1_{[0,t]},Y(u) - Y(t))  = \mathrm{Law}_Q(\tilde{\phi}_n 1_{[0,t]},Y_n(u)-{Y}_n(t))$.
	This implies also 
	$$
	\mathrm{Law}_Q(\tilde{\phi}_n 1_{[0,t]},Y_n(u) - Y_n(t))  = \mathrm{Law}_Q(\tilde{\phi}_n 1_{[0,t]}) \otimes \mathrm{Law}_Q(Y_n(u) - {Y}_n(t)).
	$$
	Passing to the limit as $n\to\infty$,
	$$
	\mathrm{Law}_Q(\phi^* 1_{[0,t]},Y^*(u) - Y^*(t))  = 
	\mathrm{Law}_Q(\phi^* 1_{[0,t]}) \otimes \mathrm{Law}_Q(Y^*(u) - {Y}^*(t)),$$
	which implies independence of $\phi^{\ddagger} 1_{[0,t]}\in L^{\beta}$ from $_t Y\in\mathcal{D}^m$ as well where
	$(_t Y)_s:= 0$ if $0\leq s\leq t$ and $(_t Y)_s:=Y_s-Y_t$, $t<s\leq 1$. 
	
	Since $Y$ is clearly a 
	measurable function 
	of $(_t Y,^t Y)\in \mathcal{D}^m\times\mathcal{D}^m$, applying Lemma \ref{fuggi} with the choice $\mathfrak{b}:=_t Y$ and $\mathfrak{a}:=(U,^t Y)$ we get that
	$\phi^{\ddagger} 1_{[0,t]}$ is $\mathcal{G}_t$-measurable, for all $t$, see Remark \ref{ipszilon}. 
	Applied to $\phi:=\phi^{\ddagger}$, Lemma \ref{versi} 
	provides a $\mathcal{G}_t$-progressively measurable $\phi^{\dagger}$ such that 
	$\mathrm{Law}_Q(\phi^*,Y^*) =\mathrm{Law}(\phi^{\dagger},Y)$, so $\phi^{\dagger}\in\mathcal{A}$ .
	
	As $\mathrm{Law}(\phi^{\dagger},Y)=\mathrm{Law}_Q
	(\phi^*,Y^*)$, also $\mathrm{Law}(\phi^{\dagger},S)=\mathrm{Law}_Q
	(\phi^*,S^*)$. Recalling Lemmata \ref{con} and \ref{cop},
	$$
	\mathrm{Law}_Q\left(-\int_0^1 S^*(t)\phi^*(t)\, dt-\int_0^1 H(S^*(t))|\phi^*(t)|^{\alpha}\, dt-W^*\right)=
	\mathrm{Law}(X_1(\phi^{\dagger})-W).
	$$
	It follows from \eqref{hoarse} that $V(X_1(\phi^{\dagger})-W)>-\infty$ and 
	$\int_0^1 \phi^{\dagger}(t)\, dt=0$ by \eqref{akkad}
	hence $\phi^{\dagger}\in\mathcal{A}'$ and $\phi^{\dagger}$
	is the maximizer we have been looking for.
\end{proof}

\begin{remark}\label{indi} {\rm The last part of the proof shows why the independent increments property of $Y$ is a crucial
hypothesis. Even if the sequence $\phi_n 1_{[0,t]}$ is $\mathcal{G}_t$-measurable this does not necessarily hold for its
limit (in any sense). So we proceed by noticing that $\tilde{\phi}_n 1_{[0,t]}$ are ``orthogonal'' to $_t Y_n$ and this property
easily passes to the limit and leads to the eventual construction of $\phi^{\dagger}$.

The need for $U$ is also apparent: taking a limit in weak convergence may easily generate additional randomness (think about the 
construction of weak solutions for stochastic equations such as the Tanaka equation) and hence $\phi^* 1_{[0,t]}$ is
not necessarily a functional of $Y^*$ (even though each $\tilde{\phi}_n$ was a functional of $Y_n$).}
\end{remark}

\begin{remark}\label{worth} {\rm It is worth commenting on the use of convex combinations in papers
		dealing with concave utility functions (e.g. \cite{KS99,ks_exp,sch}) in comparison with the 
		current paper. 
		
		When the utility function is concave and no distortions are present then using
		convex combinations improves performance (either in the utility maximization or
		in its dual problem where minimization of a convex functional is considered).
		Converging convex combinations thus directly yield an optimizer in these cases.
		
		The present setting is essentially different: the optimizer is found as
		the weak limit of a sequence of laws. At this point, taking convex combinations $\widehat{\phi}_n$
		of the $\phi_n$
		would not make sense since, by lack of convexity, $v_n:=V(X_1(\widehat{\phi}_n)-W)$ 
		will not necessarily dominate the respective convex combintions of $V(X_1({\phi}_n)-W)$ and thus $v_n$ may cease to converge to
		$\sup_{\phi} V(X_1(\phi)-W)$. 
		
		In our approach, convex combinations are needed at a subsequent stage, in order
		to show that the $X^*$ we constructed is indeed (dominated by) a portfolio value. 
		The representation of Theorem \ref{altalanos} is crucial in that argument as it permits to form
		such convex combinations on the auxiliary probability space $(O,\mathcal{H},Q)$. 
		A similar use of Skorohod's representation theorem appears
		in the proof of Lemma A.6.4 in \cite{kp} where we drew our inspiration from.} 
\end{remark}

\begin{remark}
	{\rm Theorem \ref{main} proves the existence of an optimizer in the family of
		randomized strategies, i.e. $\mathcal{G}_t$-progressively measurable ones. Such a formulation with ``relaxed'' controls is standard,
		see e.g. \cite{fleming,borkar}, but there are at least two additional arguments in favour
		of this family in our specific setting. 
		
		Lack of concavity and the presence of distortions seem to exclude arguments
		based on almost sure convergence of convex combinations (e.g. the Koml\'os lemma, see Lemma \ref{komlos1}) which is the typical
		technique for existence proofs in infinite dimensional spaces, see e.g. \cite{KS99,cs}.
		The natural way to attack such problems is switching to convergence in law (as usual in the
		weak formulation of stochastic differential equations, too), see e.g. \cite{cd,cr}. However, it was
		demonstrated in Section 5 of \cite{miki} that the set of laws of attainable
		portfolio values can easily fail to be closed for weak convergence, even in one-step, 
		frictionless models.
		Finding an optimizer over a non-closed domain looks hopeless. It was shown 
		in Section 6 of \cite{cr} in a discrete-time setting that the set of attainable
		portfolio values becomes closed when using randomized strategies. This is the first
		reason for our choice of the class of feasible strategies.
		
		As noticed in Section 5 of \cite{cr}, the investor
		may actually increase her satisfaction by randomizing her strategy, a phenomenon due to the presence of distortions.
		This is a second argument for the use of $\mathcal{G}_t$-adapted strategies.
		
		It is a delicate question under what kind of conditions Theorem \ref{main} remains
		true with $\mathcal{G}_t$-measurable strategies replaced by $\mathcal{F}_t$-measurable
		ones. This is object of current research.}
\end{remark}

\begin{remark} 
{\rm It would also be desirable to exhibit cases where optimal strategies can be found that are adapted 
to $(\mathcal{F}^S_t)_{t\in [0,1]}$, the filtration generated by the asset price $S$. 
As explained in Remark \ref{worth}, it doesn't seem possible to carry out the construction of optimal strategies 
using almost sure convergence, so even if we start with an $(\mathcal{F}^S_t)_{t\in [0,1]}$-adapted optimizer sequence 
there is no guarantee that an $(\mathcal{F}^S_t)_{t\in [0,1]}$-adapted limit could be found. 
		
One could expect positive results in the case where $(\mathcal{F}^S_t)_{t\in [0,1]}$ is ``rich enough'' in the sense that for every 
$$\mu\in\{\mathrm{Law}(X_1(\phi)):\ \phi\in\mathcal{A}'\}$$
there exists $\lambda$, adapted to $(\mathcal{F}^S_t)_{t\in [0,1]}$, such that $\mathrm{Law}(X_1(\lambda))=\mu$. 
This seems to require delicate arguments even in the case of frictionless markets and hence lies beyond the scope of the present work.} 	
\end{remark}

\begin{remark}\label{rob}
	{\rm The main lines of the above proof seem to work in situations with model uncertainty as well.
		In that case $V(X_1(\phi))=V^S(X_1(\phi))$ is calculated for each $S\in\mathcal{S}$ where $\mathcal{S}$ is
		a family of price processes and one tries to maximize the worst-case functional
		$\inf_{S\in\mathcal{S}}V^S(X_1(\phi))$ over $\phi$, see e.g. \cite{nutz,ariel}.
		
		As an optimizing sequence one can consider $(\phi^n,S_n)$ instead of $(\phi^n,Y_n)$ in this case
		(assuming e.g. that each $S\in\mathcal{S}$ generates the same filtration as $Y$). In this setting
		the $S_n$ may well have different laws. Nevertheless, under appropriate tightness conditions 
		%(and using extensions of Lemma \ref{fonction}, see \cite{freddy}) 
		the construction above may provide an optimizer $\phi^*$ as well as the worst-case model 
		$S^*=\mathrm{argmin}_S V^S(X_1(\phi^*))$. These ideas
		are left for exploration in future research.} 
\end{remark}

\section{A formulation with generalized strategies}\label{j}

In the theory of stochastic differential equations, the concept of weak solutions
allows to vary the underlying probability space, giving more flexibility for the
construction of solutions. Weak convergence of probability measures and
Skorohod's theorem are typical tools in that area.

As weak convergence techniques predominate in the present
paper as well and the solution $\phi^{\dagger}$ is also a ``weak'' one, in this section we
reformulate problem \eqref{problem} in a manner that is closer in spirit to
the world of weak solutions.

\begin{definition}\label{defi:general_stra} 
	Let $(\Omega,\mathcal{F},(\mathcal{F}_t)_{t\in [0,1]}, P)$, $Y,S,W$ be as in Section \ref{harom}. A five-tuple 
	$$ 
	\pi=(O, \mathcal{H}, Q, \tilde{Y}, \tilde{\phi})
	$$
	is called a \emph{generalized strategy} if 
	$(O,\mathcal{H},Q)$ is a probability space, $\tilde{Y}_t$, $t\in [0,1]$ is a c\`adl\`ag process on 
	$(O,\mathcal{H},Q)$ identical in law to $Y$, and
	$\tilde{\phi}$ is a process on $(O,\mathcal{H},Q)$ such that $\tilde{\phi}_s$ is independent
	of $\tilde{Y}_u-\tilde{Y}_t$ for all $0\leq s\leq t<u$ and $\int_0^1 |\tilde{\phi}_t|\, dt<\infty$
	$Q$-a.s.
\end{definition} 

Clearly, $\tilde{Y}$ is the copy of $Y$ providing the information structure on $(O,\mathcal{H},Q)$
and $\tilde{\phi}$ represents the trading speed which must be non-anticipative with respect
to $\tilde{Y}$. Defining $\tilde{W}:=\ell(\tilde{Y})$ and $\tilde{S}:=f(\tilde{Y})$,
we can set 
\begin{eqnarray*}
X_1(\tilde{\phi}) &:=& -\int_0^1 \tilde{\phi}_t \tilde{S}_t\, dt-\int_0^1
H(\tilde{S}_t) |\tilde{\phi}_t|^{\alpha}\, dt,\\
\tilde{X}_1(\tilde{\phi}) &:=& \int_0^1 \tilde{\phi}_t\, dt.
\end{eqnarray*}

Let $\Pi$ denote the class of generalized strategies. Another version of Theorem \ref{main}
could be stated as follows. Its proof closely follows that of Theorem \ref{main}.

\begin{theorem}\label{main2}
	Let Assumptions \ref{u}, \ref{modl} and \ref{bou} be in force. Let $\Pi':=\{\pi\in\Pi:V^Q_-([X_1(\tilde{\phi})
	-\tilde{W}]^-)<\infty,\ \tilde{X}_1(\phi)=0\}$. There exists $\pi^{\dagger}\in\Pi'$ such that
	$$
	V^{Q^{\dagger}}(X_1(\tilde{\phi}^{\dagger})
	-W^{\dagger}) = \sup_{\pi\in\Pi'} V^Q(X_1(\tilde{\phi})-\tilde{W}),
	$$
	where $\pi^{\dagger}=(O^{\dagger},\mathcal{H}^{\dagger},Q^{\dagger},\tilde{Y}^{\dagger},
	\tilde{\phi}^{\dagger})$ and $W^{\dagger}=\ell(\tilde{Y}^{\dagger})$.\hfill $\Box$
\end{theorem}

\section{Auxiliary results}\label{negy}

\begin{lemma}\label{fuggi}
	Let $(A,\mathcal{A})$, $(B,\mathcal{B})$ be measurable spaces and $j:A\times B\to\mathbb{R}$
	a measurable mapping. Let $(\mathfrak{a},\mathfrak{b})$ be an $A\times B$-valued random variable.
	If $\sigma(j(\mathfrak{a},\mathfrak{b}),\mathfrak{a})$ is independent of $\mathfrak{b}$ then 
	$j(\mathfrak{a},\mathfrak{b})$ is
	$\sigma(\mathfrak{a})$-measurable.
\end{lemma}
\begin{proof} Denote by $\mu_{A}(\cdot)$ (resp. $\mu_{B}(\cdot)$) the
	law of $\mathfrak{a}$ (resp. $\mathfrak{b}$).
	Considering $\arctan\circ j$ instead of $j$, we may and will assume that $j$ is bounded.
	We claim that $j(\mathfrak{a},\mathfrak{b})=k(\mathfrak{b})$ a.s. where $k(a):=\int_B j(a,b) \mu_{B}(db)$.
	Using an argument with monotone classes, it suffices to establish that, for all bounded measurable 
	$m:\,A\to \mathbb{R}$,
	$n:\,B\to\mathbb{R}$, one has $Ej(\mathfrak{a},\mathfrak{b})m(\mathfrak{a})n(\mathfrak{b})=
	Ek(\mathfrak{a})m(\mathfrak{a})n(\mathfrak{b})$.
	By independence of $\mathfrak{a},\mathfrak{b}$ and by definition, the latter expression equals 
	\begin{eqnarray*}
		Ek(\mathfrak{a})m(\mathfrak{a})En(\mathfrak{b})=\int_A \int_B j(a,w)\mu_{B}(dw) m(a)\mu_{A}(da)
		En(\mathfrak{b}) &=&\\
		\int_A \int_B j(a,w)m(a)\mu_{B}(dw)\mu_{A}(da)
		En(\mathfrak{b})=Ej(\mathfrak{a},\mathfrak{b})m(\mathfrak{a})En(\mathfrak{b}) &=& \\
		Ej(\mathfrak{a},\mathfrak{b})m(\mathfrak{a})n(\mathfrak{b}), & &
	\end{eqnarray*}
	by our independence hypothesis. This completes the proof.
\end{proof}

We now recall Th\'eor\`eme 1 of \cite{beksy}. Just like Theorem \ref{altalanos} in Section \ref{tetto}
above, this result is crucial
for the developments of the present paper.

\begin{lemma}\label{fonction} Let $A,B$ be separable metric spaces and $\xi_n\in A$, $n\in\mathbb{N}$ a sequence of 
	random variables
	converging to $\xi\in A$ in probability such that $\mathrm{Law}(\xi_n)$ is the same for all $n$.
	Then for each measurable $h:A\to B$ the random variables $h(\xi_n)$ converge to $h(\xi)$
	in probability (hence also a.s. along a subsequence).
	\hfill $\Box$
\end{lemma}

%We need to extend Theorem 6.10 of \cite{k}. 
\begin{lemma}\label{transfer} Let the topological space $Z$ be the union of its closed, increasing subspaces $A_n$, $n\in\mathbb{N}$ which 
	are Polish spaces (with appropriate metrics) and let $B$ be a measurable space.
	Let $H,\tilde{H}$ be random elements in $B$ with identical laws, defined on the probability spaces $(\Xi,\mathcal{E},R)$,
	$(\tilde{\Xi},\tilde{\mathcal{E}},\tilde{R})$, respectively. Let $\tilde{\phi}$ be a random element in $Z$, defined on $(\tilde{\Xi},\tilde{\mathcal{E}},\tilde{R})$. 
	Let $U$ be independent of $H$ with uniform law on $[0,1]$. There exists a measurable 
	function $f: B \times [0,1] \to Z$ such that $\phi = f(H, U)$ satisfies $Law_R(H, \phi) = Law_{\tilde{R}}(\tilde{H}, \tilde{\phi})$.
\end{lemma}
\begin{proof} 
In view of Lemma 3.22 of \cite{k} it suffices to show that $Z$ is a Borel space in the sense
of \cite{k}, i.e. it is
Borel-isomorphic to a Borel subset of $[0,1]$. We define $C_n := A_n \setminus \cup_{i < n}A_i$, $n\geq 0$.
Borel subsets of Polish spaces are clearly Borel spaces, let $\psi_n:C_n\to [1-1/2^{2n+1},1-1/2^{2n+2}]$
be Borel isomorphisms attesting this. Then it is easy to check that $\psi(x):=\psi_n(x)$,
$x\in C_n$ defines a Borel isomorphism between $Z$ and a Borel subset of $[0,1]$.
\end{proof}
%	 Since each $A_n$ is a Polish space, Theorem 6.3 of \cite{k} implies that there exists a 
%	kernel $\mu_n$ from $B$ to $A_n$ satisfying $\tilde{R}[\tilde{\phi} \in E \cap A_n| \tilde{H}] = \mu_n(\tilde{H}, E \cap A_n)$ a.s., for each
%	$E\in\mathcal{B}(Z)$.
%	Now, it holds that
%	$$ \tilde{R}[\tilde{\phi} \in E | \tilde{H}] = \sum_{i=1}^{\infty} \mu_i(\tilde{H}, E\cap C_i) =: \mu(\tilde{H}, E).$$
%	From Lemma 3.22 of \cite{k}, there exists a measurable 
%	$f_n: B \times [0,1] \to A_n$ such that $f_n(x, U)$ has distribution 
%	$\mu_n(x, \cdot)/\mu(x,A_n)$ for each $x\in B$ such that $\mu(x,A_n)>0$. Now, we define
%	$$f(h,u) = f_n(h,u)  \qquad \text{for } (h,u)\in f_n^{-1}(C_n).$$
%	We compute
%	\begin{align*}
%	R\left[ f(H,U) \in E  \right] &= \sum_{n} {R}\left[ f(H,U) \in E \cap C_n \right] = \sum_{n} {R}\left[ f_n(H,U) \in E \cap C_n \right]\\
%	&= \sum_{n} \mu_n (H, E \cap C_n) = \mu(H, E),
%	\end{align*}
%	Now, for any measurable function $g: B \times [0,1] \to \mathbb{R}_+$, we get
%	\begin{align*}
%	{E}_R g(H, \phi) &= {E}_R g(H,f(H,U)) = {E}_R \int_0^1 {g(H, f(H,u))du}\\
%	&= {E}_R \int_B{g(H, x) \mu(H, dx)} = {E}_{\tilde{R}}g(\tilde{H}, \tilde{\phi}),
%	\end{align*}
%	thus the proof is complete.

\begin{lemma}\label{con}
	The mapping $F:L^{\beta}\times \mathcal{D}^1\to \mathbb{R}$ defined by $F(\psi,\chi):=\int_{[0,1]} \psi(t) \chi(t)\, dt$
	is sequentially continuous and (jointly) measurable when $L^{\beta}$ is equipped with the weak topology.
\end{lemma}
\begin{proof} Take sequences $\psi_n\to \psi$ in $L^{\beta}$ and $\chi_n\to \chi$ in $\mathcal{D}^1$. 
	Then the sequence $\chi_n$, being relatively
	compact in $\mathcal{D}^1$, is uniformly bounded by a constant $K$ (see Theorem 12.3
	of \cite{billingsley}) 
	and $\chi_n(t)$ tends to $\chi(t)$ at every continuity point $t$ of the latter, in particular outside
	a countable set (see page 124 of \cite{billingsley}). So 
	\begin{eqnarray*}
		\left| \int_{[0,1]} \psi_n(t) \chi_n(t)\, dt-\int_{[0,1]} \psi(t) \chi(t)\, dt\right| &\leq&\\
		\int_{[0,1]} |\psi_n(t) \chi_n(t)-\psi_n(t) \chi(t)|\, dt+
		\int_{[0,1]} |\psi_n(t) \chi(t)-\psi(t) \chi(t)|\, dt &\leq&\\
		\left(\int_{[0,1]} |\psi_n(t)|^{\beta} dt\right)^{1/\beta}\left(\int_{[0,1]} |\chi_n(t)-\chi(t)|^{\gamma}\, dt
		\right)^{1/\gamma} &+&\\
		\int_{[0,1]} |\psi_n(t) \chi(t)-\psi(t) \chi(t)|\, dt.& &
	\end{eqnarray*}
	The first term tends to $0$ as $\psi_n$ is weakly bounded hence also norm bounded in $L^{\beta}$ and Lebesgue's theorem
	applies to $|\chi_n(t)-\chi(t)|^{\gamma}\leq (2K)^{\gamma}$. The second term tends to $0$ by the weak convergence 
	of $\psi_n$ to
	$\psi$ noting that $\chi\in L^{\gamma}$ trivially.
	
	As closed balls with radius $r$ around the origin in $L^{\beta}$ (denoted by $B_r$) are metrizable by 
	the separability of $L^{\gamma}$, sequential continuity implies
	continuity and hence measurability of $F$ restricted to $B_r\times \mathcal{D}^1$ for every $r$,
	which easily implies the measurability of $F$ on the whole of $L^{\beta}\times\mathcal{D}^1$.
\end{proof}

\begin{lemma}\label{cop}
	The mapping $(s,\phi)\in \mathcal{D}^1\times L^{\beta}\to \int_0^1 H(s(t))|\phi(t)|^{\alpha}\, dt\in\mathbb{R}$
	is $\mathcal{B}(\mathcal{D}^1\times L^{\beta})$ measurable when $L^{\beta}$ is equipped
	by the weak topology. 
\end{lemma}
\begin{proof}
	By the monotone convergence theorem it is enough to prove the measurability of  
	\begin{equation}\label{guld}
	(s,\phi)\to \int_0^1 H(s(t))(|\phi(t)|^{\alpha}\wedge N)\, dt,
	\end{equation}
	for all $N>0$. Since $L^{\beta}$ is a separable Banach space, by results of \cite{pettis}, Borel sets
	of $L^{\beta}$ for the weak topology coincide with those of the norm topology. So it suffices to
	prove continuity of \eqref{guld} when $L^{\beta}$ is equipped with the norm topology.
	Let $(s_n,\phi_n)\to (s,\phi)$ in $\mathcal{D}^1\times L^{\beta}$. Then $s_n$ are uniformly
	bounded and converge Lebesgue-a.s. to $s$ and $\phi_n$ converge to $\phi$ in Lebesgue measure.
	Dominated convergence implies the convergence of $\int_0^1 H(s_n(t))(|\phi_n(t)|^{\alpha}\wedge N)\, dt$
	to $\int_0^1 H(s(t))(|\phi(t)|^{\alpha}\wedge N)\, dt$ as $n\to\infty$.
\end{proof}

\begin{lemma}\label{versi}
	Let $\phi:\Omega\to L^{\beta}$ be such that $\sigma(\phi 1_{[0,t]})\subset\mathcal{G}_t$ for all $t$. 
	Then there exists $\bar{\phi}(\omega,t)=\phi(\omega)(t)$,
	$P\times Leb$-a.s. such that $\bar{\phi}_t$ is $\mathcal{G}_t$-progressively measurable.
\end{lemma}
\begin{proof}
	Define $$
	\check{\phi}(\omega,t):=\limsup_n n\int_{t-1/n}^t \phi(\omega)(s)ds,\quad \bar{\phi}(\omega,t):=\check{\phi}(\omega,t)1_{\{\check{\phi}(\omega,t)<\infty\}}.
	$$
	By Lebesgue's differentiation theorem and by measurability of $\omega\to\phi(\omega)\in L^{\beta}$ this is $\mathcal{F}\otimes\mathcal{B}([0,1])$-measurable
	and equals $\phi(\omega)(t)$, $P\otimes Leb$-a.s. By $\sigma(\phi 1_{[0,t]})\subset\mathcal{G}_t$ we get 
	progressive measurability, too.
\end{proof}

%The following characterization for measurability is from \cite{pettis}.

%\begin{theorem}\label{peti}
%Let $(\Omega,\mathcal{F},P)$ be a complete probability space and $\mathbb{B}$ a separable Banach space.
%Then $X:\Omega\to\mathbb{B}$ is measurable iff $g\circ X$ is measurable for all $g\in\mathbb{B}'$.\hfill $\Box$
%\end{theorem}

We recall the main result of \cite{komlos}, see the Appendix of \cite{kp} for a recent account of the proof.
\begin{lemma}\label{komlos1}
	Let $f_n$ be a sequence of real-valued random variables satisfying $$\sup_n E|f_n|<\infty.$$
	Then there is a subsequence $n_j$, $j\in\mathbb{N}$ and a random variable $f$ 
	such that
	$$
	\widehat{f}_i:=\frac{1}{i}\sum_{j=1}^{i}f_{n_j}\to f,\mbox{ a.s.},i\to\infty.
	$$ 
\end{lemma}

We will need an easy corollary of the above lemma, used in the proof of Theorem \ref{main}.

\begin{corollary}\label{komlos} Let $f_n:O\times [0,1]\to\mathbb{R}$, $n\in\mathbb{N}$
	be $\mathcal{H}\otimes\mathcal{B}([0,1])$-measurable such that 
	$$
	J:=\sup_n \int_0^1 |f_n(\omega,t)|\, dt<\infty
	$$
	almost surely. Then there is a subsequence $n_j$, $j\in\mathbb{N}$ and $f:O\times [0,1]\to\mathbb{R}$ 
	such that
	$$
	\widehat{f}_i:=\frac{1}{i}\sum_{j=1}^{i}f_{n_j}\to f,\ Q\otimes Leb\mbox{-a.s.},\ i\to\infty.
	$$ 
\end{corollary}
\begin{proof} Define $d\mu/d(Q\otimes Leb):=e^{-J}/Ee^{-J}$. Under $\mu$, Lemma \ref{komlos1} applies 
	to the sequence $f_n$ 
	so we get $\widehat{f}_i$ converging to $f$, $\mu$-a.s. Since $\mu\sim Q\otimes Leb$, this completes the proof.
\end{proof}

\section*{Acknowledgments}
The authors thank the anonymous referees for their useful comments and gratefully acknowledge the support of the ``Lend\"ulet'' grant LP 2015-6 of the
Hungarian Academy of Sciences.

\bibliographystyle{siamplain}
\bibliography{references}
\end{document}